\documentclass[reqno,dvips]{amsart}

\usepackage{mathrsfs}
\usepackage{amsmath}
\usepackage{amssymb}
\usepackage{cite}
\usepackage{latexsym}
\usepackage{graphicx}

\usepackage{color}

\theoremstyle{plain}
\newtheorem{thm}{Theorem}[section]
\newtheorem{lem}[thm]{Lemma}
\newtheorem{cor}[thm]{Corollary}
\newtheorem{dfn}[thm]{Definition}
\newtheorem{prop}[thm]{Proposition}
\newtheorem{rmk}[thm]{Remark}
\newtheorem{ex}[thm]{Example}

\def\K{\mathcal{K}}
\def\M{\mathcal{M}}

\def\d{\mathrm{d}}
\def\h{\mathrm{h}}

\def\Cset{\mathbb{C}}
\def\Kset{\mathbb{K}}
\def\Lset{\mathbb{L}}
\def\Nset{\mathbb{N}}

\def\Rset{\mathbb{R}}

\def\Zset{\mathbb{Z}}

\def\GL{\mathrm{GL}}
\def\SL{\mathrm{SL}}

\def\id{\mathrm{id}}
\def\Im{\mathrm{Im}\,}
\def\Re{\mathrm{Re}\,}

\def\epsilon{\varepsilon}

\DeclareMathOperator{\sech}{sech}

\DeclareMathOperator{\tr}{tr}

\makeatletter
 \@addtoreset{equation}{section}
\makeatother
\def\theequation{\arabic{section}.\arabic{equation}}

\begin{document}


\title[Integrability of the Zakharov-Shabat systems]%
{Integrability of the Zakharov-Shabat systems\\ 
by quadrature}

\author{Kazuyuki Yagasaki}

\address{Department of Applied Mathematics and Physics, Graduate School of Informatics,
Kyoto University, Yoshida-Honmachi, Sakyo-ku, Kyoto 606-8501, JAPAN}
\email{yagasaki@amp.i.kyoto-u.ac.jp}

\date{\today}

\begin{abstract}
We study the integrability of the general two-dimensional \emph{Zakharov-Shabat systems},
 which appear in application of the \emph{inverse scattering transform} (IST)
 to an important class of nonlinear partial differential equations (PDEs)
 called \emph{integrable systems},
 in the meaning of \emph{differential Galois theory}, i.e., their solvability by quadrature.
It becomes a key for obtaining analytical solutions to the PDEs by using the IST.
For a wide class of potentials,
 we prove that they are integrable in that meaning
 if and only if the potentials are reflectionless.
It is well known that for such potentials
 particular solutions called \emph{$n$-solitons} in the original PDEs
 are yielded by the IST.
\end{abstract}

\maketitle


\section{Introduction}

The \emph{inverse scattering transform} (IST) is a powerful tool
 to solve the initial value problems
 for an important class of nonlinear partial equations (PDEs) called \emph{integrable systems}
 such as the Korteweg-de\,Vries (KdV) equation and nonlinear Schr\"{o}dinger (NLS) equation
 \cite{A11,AKNS73,AKNS74,APT04,AS81,L80,N85}.
In application of the technique,
 eigenvalue problems for linear systems of ordinary differential equations (ODEs)
 called the \emph{Zakharov-Shabat (ZS) systems} need to be solved.
Here we are interested in the question whether their solutions can be obtained by quadrature.
Such solvability of linear ODEs can be determined
 by \emph{differential Galois theory} \cite{CH11,PS03},
 which is an extension of \emph{classical Galois theory} for algebraic equations to linear ODEs.
A linear ODE is said to be \emph{integrable} in the meaning of differential Galois theory
 if its all solutions can be obtained by quadrature.
The differential Galois theory was also utilized
 to develop a useful tool called the Morales-Ramis theory \cite{AZ10,M99,MR01}
 for determining the nonintegrability of nonlinear ODEs.
Some relations between nonintehgrability and chaotic dynamics
 in two-degree-of-freedom Hamiltonian systems
 were described in \cite{MP99,Y03,YY19} based on the Morales-Ramis theory.
Moreover, the differential Galois theory was used to discuss
 bifurcations of homoclinic orbits in four-dimensional ODEs \cite{BY12a}
 and a Sturm-Liouville problem of second-order ODEs on the infinite interval \cite{BY12b}.

In this paper we study the integrability of the two-dimensional ZS systems,
\begin{equation}
v_x=
\begin{pmatrix}
-ik & q(x)\\
-1 & ik
\end{pmatrix}v,\quad
v\in\Cset^2,
\label{eqn:ZS1}
\end{equation}
and
\begin{equation}
v_x=
\begin{pmatrix}
-ik & q(x)\\
r(x) & ik
\end{pmatrix}v,
\label{eqn:ZS2}
\end{equation}
in the meaning of differential Galois theory, i.e., their solvability by quadrature,
 where the subscript $x$ represents differentiation with respect to the variable $x$,
 and $k\in\Cset$ is a constant.
Here the independent variable $x$ is originally defined in $\Rset$
 but its domain is a little extended later.
Moreover, the \emph{potentials} $q(x),r(x)$ are assumed to
 satisfy the following condition:

\begin{enumerate}
\setlength{\leftskip}{-1.5em}
\item[\bf(A)]
The potentials $q(x),r(x)$ are holomorphic in a neighborhood $U$ of $\Rset$ in $\Cset$.
Moreover, there exist holomorphic functions $q_\pm,r_\pm:U_0\to\Cset$ such that
 $q_\pm(0),r_\pm(0)=0$ and
\[
q(x)=q_\pm(e^{\mp\lambda_\pm x}),\quad
r(x)=r_\pm(e^{\mp\lambda_\pm x})
\]
for $|\Re x|$ sufficiently large,
 where $U_0$ is a neighborhood of the origin in $\Cset$,
 $\lambda_\pm\in\Cset$ are some constants with $\Re\lambda_\pm>0$,
 and the upper or lower sign is taken simultaneously
 depending whether $\Re x>0$ or $\Re x<0$.
\end{enumerate}
For the ZS system \eqref{eqn:ZS1} condition~(A) has a meaning only for $q(x)$.
Especially, $q(x),r(x)$ tend to zero as $x\to\pm\infty$ on $\Rset$ and $q,r\in L^1(\Rset)$
 if they satisfy condition~(A).
For example, if $q(x),r(x)$ are rational functions of $e^{\lambda x}$
 for some $\lambda\in\Cset$ with $\Re\lambda>0$,
 have no singularity on $\Rset$, and $q(x),r(x)\to 0$ as $x\to\pm\infty$,
 then condition~(A) holds.
We have another class of functions satisfying condition~(A) as follows.

\begin{rmk}
\label{rmk:1a}
Let $f(\xi)$ is a second- or higher-order polynomial of $\xi\in\Rset$
 such that for some $\xi_-<\xi_+$ $f(\xi_\pm)=0$,  $f(\xi)>0$ on $[\xi_-,\xi_+]$ and
\[
f_\xi(\xi_-)>0,\quad
f_\xi(\xi_+)<0.
\]
Then there exists a heteroclinic solution $\xi^\h(x)$ to the one-dimensional ODE
\begin{equation}
\xi_x=f(\xi)
\label{eqn:rmk1a}
\end{equation}
such that $\lim_{x\to\pm\infty}\xi^\h(x)=\xi_\pm$,
 where the upper or lower sign is taken simultaneously.
Since the complexification of \eqref{eqn:rmk1a} is holomorphically equivalent
 to the linearized ODE
\[
\xi_x=f_\xi(\xi_\pm)\xi
\]
near neighborhoods of $\xi=\xi_\pm$ in $\Cset$
 $($e.g., Theorem~$5.5$ in Chapter~I of {\rm\cite{IY08})},
 we see that $q(x)=\xi_x^\h(x)$ satisfies condition~{\rm(A)}
 with $\lambda_\pm=\mp f_\xi(\xi_\pm)$. 
\end{rmk}

It is well known that the ZS systems \eqref{eqn:ZS1} and \eqref{eqn:ZS2}
 appear in application of the IST for the following fundamental and important nonlinear PDEs
 (see, e.g., \cite{AKNS73,AKNS74} or Section~1.2 of \cite{AS81}):
\begin{itemize}
\setlength{\leftskip}{-2.7em}
\item
The KdV equation
\begin{equation}
q_t+6qq_x+q_{xxx}=0;
\label{eqn:KdV}
\end{equation}
\item
The NLS equation
\begin{equation}
iq_t=q_{xx}\pm 2|q|^2q
\label{eqn:NLS}
\end{equation}
with $r=\mp q^\ast$;
\item
The modified KdV (mKdV) equation
\begin{equation}
q_t\pm 6q^2q_x+q_{xxx}=0
\label{eqn:mKdV}
\end{equation}
with $r=\mp q$;
\item
The sine-Gordon
\begin{equation}
u_{xt}=\sin u
\label{eqn:sineG}
\end{equation}
with $-q=r=\frac{1}{2}u_x$;
\item
The sinh-Gordon
\begin{equation}
u_{xt}=\sinh u
\label{eqn:sinhG}
\end{equation}
with $q=r=\frac{1}{2}u_x$.
\end{itemize}
Here $q,r$ and $u$ are assumed to depend on the time variable $t$ as well as $x$,
 and the superscript `$\ast$' represents complex conjugate.
The ZS system of the form \eqref{eqn:ZS1} appears only for the KdV equation \eqref{eqn:KdV},
 and $q,r$ and $u$ are also assumed to be real
 except for the NLS equation \eqref{eqn:NLS}.
Under the transformation $(x+t,x-t)\mapsto(x,t)$,
 the sine- and sinh-Gordon equations \eqref{eqn:sineG} and \eqref{eqn:sinhG}
 are changed to
\[
u_{tt}-u_{xx}+\sin u=0\quad\mbox{and}\quad
u_{tt}-u_{xx}+\sinh u=0,
\]
respectively, in the physical coordinate system.

Here we prove that the ZS system \eqref{eqn:ZS1} (resp. \eqref{eqn:ZS2})
 is integrable in the meaning of differential Galois theory 
 if and only if the potential $q(x)$ is
 $($resp. the potentials $q(x),r(x)$ are$)$ \emph{reflectionless}.
See Section~2 for the precise statement of the result
 along with the definition of reflectionless potentials.
As stated above, the integrability of \eqref{eqn:ZS1} and \eqref{eqn:ZS2}
 in that meaning implies that they are solved by quadrature.
See Section~3 for its more precise definition.
It is also well known for the above five examples that
 when the potentials are reflectionless,
 the ZS systems \eqref{eqn:ZS1} and \eqref{eqn:ZS2} are analytically solved 
 and particular solutions called \emph{$n$-solitons} in the original PDEs
 are yielded from the potentials by the IST \cite{A11,APT04} (see also Section~4).
Our result means that the ZS systems \eqref{eqn:ZS1} and \eqref{eqn:ZS2}
 are solved by quadrature
 only in such a case under condition~(A).
The ZS system~\eqref{eqn:ZS1} is transformed to a linear Schr\"odinger equation,
 as stated in Section~2.
Its integrability in the meaning of differential Galois theory
 was also discussed in \cite{AMW11} for several classes of potentials
 which do not necessarily satisfy condition~(A)
 and in \cite{BY12b} for a special potential which satisfies condition~(A).

The outline of this paper is as follows:
In Section~2 we state our main results along with necessary terminologies and setting,
 and give two examples to illustrate the results.
We provide necessary information on differential Galois theory in Section~3
 and on scattering coefficients and reflectionless potentials in Section~4.
We also need some relations on scattering and reflection coefficients
 between the ZS system~\eqref{eqn:ZS1} and the corresponding linear Schr\"odinger equation,
 which are given in Appendix~A.
We prove the main theorems in Sections~5 and 6.


\section{Main Results}
In this section we give our main results.
Following the standard theory of the ISF (e.g., \cite{A11,APT04}) with slight modifications,
 we first define some necessary terminologies for its statement.
 
Assume that $k\neq 0$.
Taking $x\to\pm\infty$ in \eqref{eqn:ZS1} and \eqref{eqn:ZS2}, we have
\begin{equation}
v_x=
\begin{pmatrix}
-ik & 0\\
r_0 & ik
\end{pmatrix}v,
\label{eqn:ZS0}
\end{equation}
where $r_0=-1$ in \eqref{eqn:ZS1} and $r_0=0$ in \eqref{eqn:ZS2}.
Equation~\eqref{eqn:ZS0} has
\begin{equation}
\Phi(x;k)=T\begin{pmatrix}
e^{-ikx} & 0\\
0 & e^{ikx}
\end{pmatrix}T^{-1}
\label{eqn:Phi}
\end{equation}
as a fundamental matrix such that $\Phi(0)=\id_2$,
 where $\id_2$ denotes the  $2\times 2$ identity matrix and
\[
T=\begin{pmatrix}
1 & 0\\[1ex]
\displaystyle\frac{ir_0}{2k}
 & 1
\end{pmatrix}.
\]
Let $v=\phi(x;k),\bar{\phi}(x;k),\psi(x;k),\bar{\psi}(x;k)$
 be solutions to  \eqref{eqn:ZS1} or \eqref{eqn:ZS2} such that
\begin{equation}
\begin{split}
&
\phi(x;k)\sim T
\begin{pmatrix}
1\\
0
\end{pmatrix}e^{-ikx},\quad
\bar{\phi}(x;k)\sim T
\begin{pmatrix}
0\\
1
\end{pmatrix}e^{ikx}\quad\mbox{as $x\to-\infty$},\\
&
\psi(x;k)\sim T
\begin{pmatrix}
0\\
1
\end{pmatrix}e^{ikx},\quad
\bar{\psi}(x;k)\sim T
\begin{pmatrix}
1\\
0
\end{pmatrix}e^{-ikx}\quad\mbox{as $x\to+\infty$.}
\end{split}
\label{eqn:bc}
\end{equation}
These solutions are called the \emph{Jost solutions}
 and their existence for $k\neq 0$ will be shortly shown in Section~6.
Since $v=\psi(x;k),\bar{\psi}(x;k)$ are linearly independent solutions
 to \eqref{eqn:ZS1} and \eqref{eqn:ZS2},
 there exist constants $a(k),\bar{a}(k),b(k),\bar{b}(k)$ such that
\begin{equation}
\begin{split}
\phi(x;k)=& b(k)\psi(x;k)+a(k)\bar{\psi}(x;k),\\
\bar{\phi}(x;k)=&\bar{a}(k)\psi(x;k)+\bar{b}(k)\bar{\psi}(x;k).
\end{split}
\label{eqn:ab1}
\end{equation}
We refer to the constants $a(k),\bar{a}(k),b(k),\bar{b}(k)$ as \emph{scattering coefficients}.
When $a(k),\bar{a}(k)\neq 0$, the constants
\[
\rho(k)=b(k)/a(k),\quad
\bar{\rho}(k)=\bar{b}(k)/\bar{a}(k)
\]
are defined and called the \emph{reflection coefficients} for \eqref{eqn:ZS1} and \eqref{eqn:ZS2}.
If $\rho(k),\bar{\rho}(k)=0$ for any $k\in\Rset\setminus\{0\}$,
 then $q(x),r(x)$ are called \emph{reflectionless potentials}.

In the standard IST for the KdV equation \eqref{eqn:KdV},
 the linear Schr\"{o}dinger equation
\begin{equation}
w_{xx}+(k^2+q)w=0
\label{eqn:ZS1a}
\end{equation}
 is used instead of \eqref{eqn:ZS1}
 and the scattering and reflection coefficients are defined for \eqref{eqn:ZS1a}.
See Appendix~A for the relations on these  coefficients
 between \eqref{eqn:ZS1} and \eqref{eqn:ZS1a}.
Note that  the first and second component of $v$ in \eqref{eqn:ZS1} are given by
\begin{equation}
v_1=-w_{x}+ikw,\quad
v_2=w
\label{eqn:v}
\end{equation}
in \eqref{eqn:ZS1}.

We now state the first of our main results.

\begin{thm}
\label{thm:2a}
Suppose that $q(x)$ is $($resp. $q(x),r(x)$ are$)$  reflectionless
 and satisfies $($resp. satisfy$)$ condition~{\rm(A)}.
Then $q(x)$ is a rational function
 $($resp. $q(x),r(x)$ are  rational functions$)$ of $e^{\lambda x}$,
 where $\lambda\in\Cset$ is some constant with $\Re\lambda>0$.
Moreover, the ZS system \eqref{eqn:ZS1} $($resp. \eqref{eqn:ZS2}$)$,
 which is regarded as a linear system of differential equations over $\Cset(e^{\lambda x})$,
 is integrable in the meaning of differential Galois theory,
 i.e., it is solved by quadrature, for $k\in\Cset\setminus\{0\}$.
\end{thm}

Theorem~\ref{thm:2a} is proved in Section~5.

\begin{rmk}
For \eqref{eqn:ZS1}
 we can take $\lambda\in\Rset$ in Theorem~$\ref{thm:2a}$.
See Section~$4.2.1$.
\end{rmk}

\begin{figure}[t]
\includegraphics[scale=0.5]{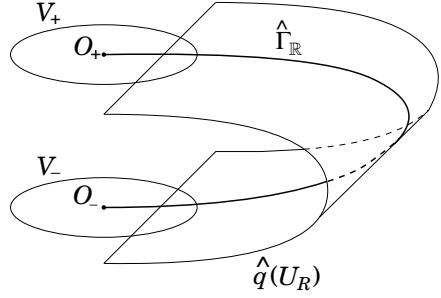}
\caption{Riemann surface $\hat{\Gamma}=\hat{q}(U_R)\cup V_+\cup V_-$.
\label{fig:2a}}
\end{figure}

Let $p(x)$ be a holomorphic function in a neighborhood $U_0$ of $\Rset$ in $\Cset$
 such that $p(x)\to\pm 1$ and $p_x(x)/q_x(x)\to 0$ as $x\to\pm\infty$,
 and let $\hat{\Gamma}_\Rset=\{\hat{q}(x)=(q(x),p(x))\mid x\in\Rset\}\cup\{(0,1),(0,-1)\}$.
Let $U_\pm$ be neighborhoods of $O_\pm=(0,\pm 1)$ in $\Cset^2$
 and let $V_\pm=U_\pm\cap\hat{q}(U_0)$.
Let $R>0$ be sufficient large
 and let $U_R\subset U_0$ be a neighborhood of the open interval $(-R,R)$ in $\Cset$
 such that $\hat{q}(U_R)$ does not contain $O_\pm$ but intersect $V_\pm$.
Thus, we define a Riemann surface $\hat{\Gamma}$ that consists of $V_\pm$ and $\hat{q}(U)$:
 $s_\pm=e^{\mp\lambda_\pm x}$ are used as the coordinates in $V_\pm$
 while the original complex variable $x\in U_R$ is used as the coordinate in $\hat{q}(U_R)$.
Note that $\hat{\Gamma}\supset\Gamma_\Rset$.
See Fig.~\ref{fig:2a}.

Let $A(x)$ be the coefficient matrix in \eqref{eqn:ZS1} or \eqref{eqn:ZS2}, i.e.,
\[
A(x)=\begin{pmatrix}
-ik & q(x)\\
-1 & ik
\end{pmatrix}
\quad\mbox{or}\quad
\begin{pmatrix}
-ik & q(x)\\
r(x) & ik
\end{pmatrix}.
\]
We express the ZS systems \eqref{eqn:ZS1} and \eqref{eqn:ZS2} as
\begin{equation}
\frac{\d\eta}{\d x}=A(x)\eta
\label{eqn:ZSU}
\end{equation}
in $\hat{q}(U_R)$, and
\begin{equation}
\frac{\d\eta}{\d s_\pm}=\mp\frac{1}{\lambda_\pm s_\pm}A_\pm(s_\pm)\eta,
\label{eqn:ZSpm}
\end{equation}
in $V_\pm$, where
\[
A_\pm(s_\pm)=\begin{pmatrix}
-ik & q_\pm(s_\pm)\\
-1 & ik
\end{pmatrix}
\quad\mbox{or}\quad
\begin{pmatrix}
-ik & q_\pm(s_\pm)\\
r_\pm(s_\pm) & ik
\end{pmatrix}
\]
for \eqref{eqn:ZS1} or \eqref{eqn:ZS2}.
Note that $s_\pm=0$ at $O_\pm$
 and $\d/\d x=\mp\lambda_\pm s_\pm\d/\d s_\pm$ in $V_\pm$.
Thus, we can regard them as a linear system of differential equations
 on the Riemann surface $\hat{\Gamma}$.

\begin{thm}
\label{thm:2b}
Suppose that $q(x)$ satisfies $($resp. $q(x),r(x)$ satisfy$)$ condition~{\rm(A)}.
If the ZS system \eqref{eqn:ZS1} $($resp. \eqref{eqn:ZS2}$)$ is integrable
 in the meaning of differential Galois theory for $k\in\Rset\setminus\{0\}$
 when it is regarded as a linear system of differential equations
 on the Riemann surface $\hat{\Gamma}$,
 then $q(x)$ is $($resp. $q(x),r(x)$ are$)$ reflectionless.
\end{thm}

Theorem~\ref{thm:2b} is proved in Section~6.

Assume that $q(x),r(x)$ are rational functions of $e^{\lambda x}$
 for some $\lambda\in\Cset$ with $\Re\lambda>0$,
 have no singularity on $\Rset$, and $q(x),r(x)\to 0$ as $x\to\pm\infty$.
Then $q(x),r(x)$ satisfy condition~(A), as stated in Section~1.
Moreover, the ZS systems \eqref{eqn:ZS1} and \eqref{eqn:ZS2}
 are regarded as linear systems of differential equations over $\Cset(e^{\lambda x})$,
 as in Theorem~\ref{thm:2a}.
In this situation we immediate obtain the following result
 as a corollary for Theorem~\ref{thm:2b}.

\begin{cor}
\label{cor:2a}
Suppose that $q(x)$ is a rational function $($resp. $q(x),r(x)$ are rational functions$)$
 of $e^{\lambda x}$ for some $\lambda\in\Cset$ with $\Re\lambda>0$,
 has $($resp. have$)$ no singularity on $\Rset$,
 and $q(x)\to 0$ $($resp. $q(x),r(x)\to 0)$ as $x\to\pm\infty$.
If the ZS system \eqref{eqn:ZS1} $($resp. \eqref{eqn:ZS2}$)$ over $\Cset(e^{\lambda x})$
 is integrable in the meaning of differential Galois theory for $k\in\Rset\setminus\{0\}$,
 then $q(x)$ is $($resp. $q(x),r(x)$ are$)$ reflectionless.
\end{cor}

\begin{proof}
If the ZS system \eqref{eqn:ZS1} or \eqref{eqn:ZS2} over $\Cset(e^{\lambda x})$
 is integrable in the meaning of differential Galois theory for $k\in\Rset\setminus\{0\}$,
 then so is it as a linear system of differential equations
 on the Riemann surface $\hat{\Gamma}$.
This yields the desired result.
\end{proof}

In closing this section, we give two examples for the ZS system~\eqref{eqn:ZS1}.
They are immediately modified as those for \eqref{eqn:ZS2}.

\begin{ex}
\label{ex:2a}
Let $q(x)=\alpha\sech^2 x$ for some $\alpha>0$.
Obviously, $q(x)$ satisfies condition~{\rm(A)}.
As shown in Section~$2.5$ of {\rm\cite{L80}},
 it is a reflectionless potential in the linear Schr\"odinger equation \eqref{eqn:ZS1a}
 and consequently in the ZS system~\eqref{eqn:ZS1}
 $($see Appendix~A, especially Eq.~\eqref{eqn:a1c}$)$
 if and only if $\alpha=n(n+1)$ for some $n\in\Nset$.
Using Theorem~$\ref{thm:2a}$ and Corollary~$\ref{cor:2a}$, we see that
 the ZS system~\eqref{eqn:ZS1} over $\Cset(e^{2x})$ is integrable
 in the meaning of differential Galois theory for $k\in\Rset\setminus\{0\}$
 if and only if $\alpha=n(n+1)$ for some $n\in\Nset$.
\end{ex}

\begin{ex}
\label{ex:2b}
Let $\alpha>1$ be a real number and let $\xi^\h(x)$ be a heteroclinic orbit in
\begin{equation}
\xi_x=\xi(\xi-1)(\xi-\alpha)
\label{eqn:ex2b}
\end{equation}
and connect $\xi=0$ to $\xi=1$.
We see that $\xi=\xi^\h(x)$ satisfies
\[
\frac{\xi^{\alpha-1}(\alpha-\xi)}{(1-\xi)^\alpha}
=\frac{\xi^\h(0)^{\alpha-1}(\alpha-\xi^\h(0))}{(1-\xi^\h(0))^\alpha}e^x,
\]
but it is difficult to obtain its closed expression.
Let $q(x)=\xi_x^\h(x)$, as in Remark~$\ref{rmk:1a}$,
Then $q(x)$ satisfies condition~{\rm(A)} with $\lambda_-=\alpha$ and $\lambda_+=\alpha-1$.
Assume that $\alpha$ and $\alpha-1$ are rationally independent.
Then it follows from Theorem~$\ref{thm:2b}$ that
 the ZS system \eqref{eqn:ZS1} on the Riemann surface $\hat{\Gamma}$
 is not integrable in the meaning of differential Galois theory for all $k\in\Rset\setminus\{0\}$
 since $q(x)$ is not a rational function of some exponential function
 and it is not reflectionless by Theorem~$\ref{thm:2a}$.
\end{ex}


\section{Differential Galois Theory}

In this and the next sections we give some prerequisites for our result.
We begin with the differential Galois theory for linear differential equations,
 which is often referred to as the Picard-Vessiot theory,
 containing monodromy groups and Fuchsian equations.
See the textbooks \cite{CH11,PS03} for more details on the theory.

\subsection{Picard-Vessiot extensions}
Consider a linear system of differential equations
\begin{equation}\label{LinearSystem}
y'=Ay,\quad A\in\mathrm{gl}(n,\Kset),
\end{equation}
where $\Kset$ is a differential field and
 $\mathrm{gl}(n,\Kset)$ denotes the ring of $n\times n$ matrices
 with entries in $\Kset$.
We recall that a \emph{differential field} is a field
 endowed with a derivation $\partial$,
 which is an additive endomorphism
 satisfying the Leibniz rule.
By abuse of notation we write $y'$ instead of $\partial y$.
The set $\mathrm{C}_{\Kset}$ of elements of $\Kset$ for which $\partial$ vanishes
 is a subfield of $\Kset$
and called the \emph{field of constants of $\Kset$}.
In our application of the theory in this paper,
 the differential field $\Kset$ is
 the field of meromorphic functions on a Riemann surface $\Gamma$,
 so that the field of constants is $\Cset$.

A \emph{differential field extension} $\Lset\supset \Kset$
 is a field extension such that $\Lset$ is also a differential field
 and the derivations on $\Lset$ and $\Kset$ coincide on $\Kset$.
A differential field extension $\Lset\supset \Kset$
 satisfying the following three conditions is called a \emph{Picard-Vessiot extension}
 for \eqref{LinearSystem}:
\begin{enumerate}
\item[\bf (PV1)]
There exists a fundamental matrix $\Xi(x)$ of \eqref{LinearSystem} with entries in $\Lset$;
\item[\bf (PV2)]
The field $\Lset$ is generated by $\Kset$ and entries of the fundamental matrix $\Xi(x)$;
\item[\bf (PV3)]
The fields of constants for $\Lset$ and $\Kset$ coincide.
\end{enumerate}
The system \eqref{LinearSystem}
 admits a Picard-Vessiot extension which is unique up to isomorphism.
We give some notions on differential field extensions.

\begin{dfn}
A differential field extension $\Lset\supset\Kset$ is called 
\begin{enumerate}
\setlength{\leftskip}{-1.8em}
\item[(i)]
an \emph{integral extension} if there exists $a\in\Lset$ such
that $a'\in \Kset$ and $\Lset = \Kset(a)$,
 where $\Kset(a)$ is the smallest extension of $\Kset$ containing $a$;
\item[(ii)]
an \emph{exponential extension} if there exists $a\in\Lset$ such
that $a'/a\in \Kset$ and $\Lset = \Kset(a)$;
\item[(iii)]
an \emph{algebraic extension} if there exists $a\in\Lset$
 such that it is algebraic over $\Kset$ and $\Lset = \Kset(a)$.
\end{enumerate}
\end{dfn}

\begin{dfn}
A differential field extension $\Lset\supset\Kset$ is called a 
\emph{Liouvillian extension} if it can be decomposed as a tower of extensions,
\[
\Lset = \Kset_n \supset \ldots \supset \Kset_1\supset 
\Kset_0 = \Kset,
\]
such that each extension $\Kset_{j+1}\supset \Kset_j$
 is either integral, exponential or algebraic. 
\end{dfn}

Thus, if the Picard-Vessiot extension $\Lset\supset\Kset$ is Liouvillian,
 then Eq.~\eqref{LinearSystem} is solved by quadrature.

We now fix a Picard-Vessiot extension $\Lset\supset \Kset$
 and fundamental matrix $\Phi$ with entries in $\Lset$
 for \eqref{LinearSystem}.
Let $\sigma$ be a \emph{$\Kset$-automorphism} of $\Lset$,
 which is a field automorphism of $\Lset$
 that commutes with the derivation of $\Lset$
 and leaves $\Kset$ pointwise fixed.
Obviously, $\sigma(\Phi)$ is also a fundamental matrix of \eqref{LinearSystem}
 and consequently there is a matrix $M_\sigma$ with constant entries
 such that $\sigma(\Phi)=\Phi M_\sigma$.
This relation gives a faithful representation
 of the group of $\Kset$-automorphisms of $\Lset$
 on the general linear group as
\[
R\colon \mathrm{Aut}_{\Kset}(\Lset)\to\GL(n,\mathrm{C}_{\Lset}),
\quad \sigma\mapsto M_{\sigma},
\]
where $\GL(n,\mathrm{C}_{\Lset})$
is the group of $n\times n$ invertible matrices with entries in $\mathrm{C}_{\Lset}$.
The image of $R$
 is a linear algebraic subgroup of $\GL(n,\mathrm{C}_{\Lset})$,
 which is called the \emph{differential Galois group} of \eqref{LinearSystem}
and denoted by $\mathrm{Gal}(\Lset/\Kset)$.
This representation is not unique
 and depends on the choice of the fundamental matrix $\Phi$,
 but a different fundamental matrix only gives rise to a conjugated representation.
Thus, the differential Galois group is unique up to conjugation
 as an algebraic subgroup of the general linear group.

Let $G\subset\GL(n,\mathrm{C}_{\Lset})$ be an algebraic group.
Then it contains a unique maximal connected algebraic subgroup $G^0$,
 which is called the \emph{connected component of the identity}
 or \emph{connected identity component}.
The connected identity component $G^0\subset G$
 is a normal algebraic subgroup and the smallest subgroup of finite index,
 i.e., the quotient group $G/G^0$ is finite.
By the Lie-Kolchin Theorem \cite{CH11,PS03},
a connected solvable linear algebraic group is triangularizable.
Here a subgroup of $\GL(n,\mathrm{C}_{\Lset})$ is said to be \emph{triangularizable}
 if it is conjugated to a subgroup of the group of (lower) triangular matrices.
 The following theorem relates
 the solvability of the differential Galois group
 with a Liouvillian Picard-Vessiot extension (see \cite{CH11,PS03} for the proof).

\begin{thm}
\label{thm:dg}
Let $\Lset\supset\Kset$ be a Picard-Vessiot extension of \eqref{LinearSystem}.
The connected identity component
 of the differential Galois group $\mathrm{Gal}(\Lset/\Kset)$ is solvable
 if and only if the extension $\Lset\supset\Kset$ is Liouvillian.
\end{thm}

Thus, if the connected identity component
 of the differential Galois group $\mathrm{Gal}(\Lset/\Kset)$ is solvable,
 then Eq.~\eqref{LinearSystem} is solved by quadrature
 and called \emph{integrable in the meaning of differential Galois theory}. 

\subsection{Monodromy groups and Fuchsian equations}

Let $\Kset$ be the field of meromorphic functions on a Riemann surface $\Gamma$.
So the set of singularities in the entries of $A=A(x)$ is a discrete subset of $\Gamma$,
 which is denoted by $S$.
We also refer to a singularity of the entries of $A(x)$ as that of \eqref{LinearSystem}.
Let $x_0\in\Gamma\setminus S$.
We prolong the fundamental matrix $\Xi(x)$ analytically
 along any loop $\gamma$ based at $x_0$ and containing no singular point,
 and obtain another fundamental matrix $\gamma\ast\Xi(x)$.
So there exists a constant nonsingular matrix $M_{[\gamma]}$ such that
\begin{equation}
\gamma\ast\Xi(x) = \Xi(x)M_{[\gamma]}.
\label{eqn:defM}
\end{equation}
The matrix $M_{[\gamma]}$ depends on the homotopy class $[\gamma]$
 of the loop $\gamma$
 and it is called the \emph{monodromy matrix} of $[\gamma]$.

Let $\pi_1(\Gamma\setminus S,x_0)$
 be the fundamental group of homotopy classes of loops based at $x_0$.
We have a representation
\[
\tilde{R}\colon \pi_1(\Gamma\setminus S,x_0)\to {\rm GL}(n,\Cset), 
\quad [\gamma]\mapsto M_{[\gamma]}.
\]
The image of $\tilde{R}$ is called the \emph{monodromy group}
 of \eqref{LinearSystem}.
As in the differential Galois group,
the representation $\tilde{R}$ depends on the choice of the fundamental matrix,
but the monodromy group is defined as a group of matrices up to conjugation.
In general, a monodromy transformation
 defines an automorphism of the corresponding Picard-Vessiot extension.
We also just write $M_\gamma$ for $M_{[\gamma]}$ below.

A singular point $x = \bar{x}$ of \eqref{LinearSystem} is called \emph{regular}
 if for any sector $a<\arg(x-\bar{x})<b$ with $a<b$
 there exists a fundamental matrix $\Xi(x) = (\Xi_{ij}(x))$
 such that for some $c > 0$ and integer $N$,
 $|\Xi_{ij}(x)|<c|x-\bar{x}|^N$ as $x\to\bar{x}$ in the sector;
 otherwise it is called \emph{irregular}.
Especially, if $A(x)=B(x)/x$, where $B(x)$ is a holomorphic at $x=0$,
 then Eq.~\eqref{LinearSystem} has a regular singularity at $x=0$
 (see, e.g., Section~2.4 of \cite{B00}).
We have the following result, which plays an essential role
 in the proof of Theorem~\ref{thm:2b} in Section~6
 (see, e.g., Theorem 5.8 in \cite{PS03} for the proof).

\begin{thm}[Schlessinger]\label{thm:sl}
Suppose that Eq.~\eqref{LinearSystem} is Fuchsian.
Then the differential Galois group of \eqref{LinearSystem}
 is the Zariski closure of the monodromy group.
\end{thm}

Assume that Eq.~\eqref{LinearSystem} is Fuchsian and $\tr A(x)=0$.
Then we have
\[
(\det\Xi(x))'=\tr A(x)\det\Xi(x)=0.
\]
Hence,  by \eqref{eqn:defM}, $\det\Xi(x)=\det\Xi(x)\det M_\gamma$, which yields
\[
\det M_\gamma=1
\]
since $\det\Xi(x)\neq 0$.
This means by Theorem~\ref{thm:sl}
 that $\mathrm{Gal}(\Lset/\Kset)\subset\SL(n,\Cset)$.
For $n=2$ we can classify such algebraic groups as follows
 (see Section~2.1 of \cite{M99} for a proof).

\begin{prop}
\label{prop:3a}
Any algebraic group $G\subset\SL(2,\Cset)$ is similar to one of the following types:
\begin{enumerate}
\setlength{\leftskip}{-1.2em}
\item[(i)] $G$ is finite and $G^0= \{\id_2\}$;
\item[(ii)] $G = \left\{
\begin{pmatrix}
\lambda & 0\\
\mu & \lambda^{-1} 
\end{pmatrix}
\middle|\,\lambda\text{ is a root of $1$, $\mu\in\Cset$}
\right\}$
and $G^0 = \left\{\begin{pmatrix}
1&0\\
\mu & 1 
\end{pmatrix}
\middle|\,\mu \in \Cset\right\}$;
\item[(iii)] 
$G = G^0 = 
\left\{
\begin{pmatrix}
\lambda &0 \\
0 & \lambda ^{-1}
\end{pmatrix}
\middle|\,\lambda \in \Cset^{*}
\right\}$;
\item[(iv)] $G = \left\{
\begin{pmatrix}
\lambda & 0\\
0 & \lambda^{-1} 
\end{pmatrix},
\begin{pmatrix}
0 & -\beta^{-1}\\
\beta & 0 
\end{pmatrix}
\middle|\,\lambda, \beta \in \Cset^{*}
\right\}$
and $G^0 = \left\{\begin{pmatrix}
\lambda &0\\
0 & \lambda^{-1}
\end{pmatrix}
\middle|\, \lambda \in \Cset^*\right\}$;
\item[(v)] $G = G^0 = \left\{
\begin{pmatrix}
\lambda & 0\\
\mu & \lambda^{-1}
\end{pmatrix}
\middle|\,\lambda \in \Cset^{*},\, \mu \in \Cset
\right\}$;
\item[(vi)] $G = G^0 = \SL (2,\, \Cset)$.
\end{enumerate}
\end{prop}

This proposition also plays a key role
 in the proof of Theorem~\ref{thm:2b} in Section~6.
  

\section{Scattering Coefficients and Refelectionless Potentials}

We next give necessary information on scattering coefficients
 and refelectionless potentials defined in Section~2.
See the textbooks \cite{A11,APT04} for more details on these materials.
 
\subsection{Scattering coefficients}
We first present some properties of the scattering coefficients.
Noting that the trace of the coefficient matrices in \eqref{eqn:ZS1} and \eqref{eqn:ZS2} are zero,
 we see by \eqref{eqn:bc} that the Wronskian of $\phi(x)$ and $\bar{\phi}(x)$
 (resp. of $\psi(x)$ and $\bar{\psi}(x)$) is one, i.e.,
\begin{equation}
\det(\phi(x;k),\bar{\phi}(x;k))=\det(\bar{\psi}(x;k),\psi(x;k))=1.
\label{eqn:W}
\end{equation}
Hence, it follows from \eqref{eqn:ab1} that
\begin{equation}
a(k)\bar{a}(k)-b(k)\bar{b}(k)=1.
\label{eqn:det}
\end{equation}
Moreover, under the transformation $x\mapsto kx$,
 the ZS systems \eqref{eqn:ZS1} and \eqref{eqn:ZS2} are rewritten as
\[
v_x=
\begin{pmatrix}
-i & \epsilon q(x)\\
\epsilon & i
\end{pmatrix}v
\]
and
\[
v_x=
\begin{pmatrix}
-i & \epsilon q(x)\\
\epsilon r(x) & i
\end{pmatrix}v,
\]
respectively, where $\epsilon=1/k$.
This means that
\begin{equation}
a(k),\bar{a}(k)\to 1,\quad
b(k),\bar{b}(k)\to 0\qquad
\mbox{as $k\to\pm\infty$}
\label{eqn:ab3}
\end{equation}
We also have the following analyticity of the  scattering coefficients.

\begin{prop}
\label{prop:4a}\
\begin{enumerate}
\setlength{\leftskip}{-1.8em}
\item[(i)]
$a(k),\bar{a}(k),b(k),\bar{b}(k)$ are analytic in $\Rset\setminus\{0\}$.
\item[(ii)]
$a(k)$ and $\bar{a}(k)$
 can be analytically continued in the upper and lower $k$-planes, respectively.
\end{enumerate}
\end{prop}

\begin{proof}
By \eqref{eqn:ab1} and \eqref{eqn:W} we have
\begin{align*}
&
a(k)=\det(\phi(x;k),\psi(x;k)),\quad
\bar{a}(k)=\det(\bar{\psi}(x;k), \bar{\phi}(x;k)),\\
&
b(k)=\det(\bar{\psi}(x;k),\phi(x;k)),\quad
\bar{b}(k)=\det(\bar{\phi}(x;k),\psi(x;k)).
\end{align*}
Since $\phi(x;k),\bar{\phi}(x;k),\psi(x;k),\bar{\psi}(x;k)$ are bounded
 and analytic in $k\in\Rset\setminus\{0\}$, we obtain part~(i).
See Section~ 9.2 of \cite{A11} and Section~2.2.2 of \cite{APT04}
 for a proof of part~(ii).
Here we notice \eqref{eqn:a1b}.
\end{proof}

\begin{rmk}
\label{rmk:4a}\
\begin{enumerate}
\setlength{\leftskip}{-1.8em}
\item[(i)]The ZS system \eqref{eqn:ZS2}
 has the Jost solutions satisfying \eqref{eqn:bc} for $k=0$,
 so that the scattering coefficients are still defined and analytic at $k=0$.
\item[(ii)]
It follows by the identity theorem $($e.g., Theorem~$3.2.6$ of {\rm\cite{AF03}}$)$
 from Proposition~$\ref{prop:4a}$ and \eqref{eqn:ab3}
 that zeros of $a(k),\bar{a}(k)$ are isolated and their numbers are finite.
\end{enumerate}
\end{rmk}

\subsection{Reflectionless potentials}

We now assume that $b(k),\bar{b}(k)=0$ for $k\in\Rset\setminus\{0\}$,
 i.e., $q(x)$ and $q(x),r(x)$ are reflectionless potentials
 in \eqref{eqn:ZS1} and \eqref{eqn:ZS2}, respectively.
Note that $a(k),\bar{a}(k)\neq 0$ for $k\in\Rset\setminus\{0\}$, by \eqref{eqn:det}.
For \eqref{eqn:ZS1} and \eqref{eqn:ZS2} separately,
 we provide some formulas for reflectionless potentials and Jost solutions.
 
\subsubsection{ZS system \eqref{eqn:ZS1}}
We begin with the ZS system \eqref{eqn:ZS1}.
Following the standard IST theory for the KdV equation \eqref{eqn:KdV}
 (e.g., Chapter~9 of \cite{A11}),
 we discuss the linear Schr\"odinger equation \eqref{eqn:ZS1a} instead of \eqref{eqn:ZS1}.
Let $\hat{a}(k),\hat{b}(k)$ be the scattering coefficients for \eqref{eqn:ZS1a}, as in Appendix~A.
Note that $\hat{b}(k)=0$ for $k\in\Rset\setminus\{0\}$ by \eqref{eqn:a1b}.

Suppose that $\hat{a}(k)$ has $n$ simple zeros $\{k_j\}_{j=1}^n$,
 where $\Im k_j>0$.
Let
\[
\hat{N}_j(x)=\hat{\psi}(x;k_j)e^{ik_jx}
\]
for $j=1,\ldots,n$,
 where $\hat{\psi}(x;k)$ is the Jost solution to \eqref{eqn:ZS1a} satisfying \eqref{eqn:bc1a}.
Note that by \eqref{eqn:bc1a} and \eqref{eqn:ab1a}
\begin{equation}
\begin{split}
&
\hat{N}_j(x)\sim e^{2ik_jx}\quad\mbox{as $x\to+\infty$},\\
&
\hat{N}_j(x)\sim \hat{b}(k_j)^{-1}\quad\mbox{as $x\to-\infty$}
\end{split}
\label{eqn:N1a}
\end{equation}
since $\Im k_j>0$, $\hat{a}(k_j)=0$ and $\hat{b}(k_j)\neq 0$ by \eqref{eqn:a1b}.
Then we can show that they satisfy
\begin{equation}
\hat{N}_\ell(x)
=e^{2ik_\ell x}\biggl(
 1-\sum_{j=1}^n\frac{\hat{C}_j\hat{N}_j(x)}{k_\ell+k_j}\biggr),\quad
\ell=1,\ldots,n,
\label{eqn:N1}
\end{equation}
where
\[
\hat{C}_j=\frac{\hat{b}(k_j)}{\hat{a}_k(k_j)},\quad
j=1,\ldots,n.
\]
Moreover, we have
\begin{equation}
q(x)=\frac{\partial}{\partial x}\biggl(2i\sum_{j=1}^n\hat{C}_j \hat{N}_j(x)\biggr)
\label{eqn:q1}
\end{equation}
and
\begin{equation}
\hat{\psi}(x;k)=\biggl(1-\sum_{j=1}^n\frac{\hat{C}_j\hat{N}_j(x)}{k+k_j}\biggr)e^{ikx}
\label{eqn:ZS1sol}
\end{equation}
for $k\in\Cset$.
See, e.g., Sections~9.1-9.3 of \cite{APT04} for the derivations of the above relations.
Since they are obtained by the basic arithmetic operations from \eqref{eqn:N1},
 we see that $\hat{N}_\ell(x)$, $\ell=1,\ldots,n$,
 are rational functions of $e^{2ik_jx}$, $j=1,\ldots,n$.
It follows from \eqref{eqn:N1a} and \eqref{eqn:q1} that
\[
\lim_{x\to\pm\infty}q(x)=0.
\]
From the standard IST theory we see that $k_j$, $j=1,\ldots,n$, are purely imaginary
 in the upper half complex plane.
Moreover, in the KdV equation \eqref{eqn:KdV},
 Eq.~\eqref{eqn:q1} corresponds to an initial condition of an $n$-soliton.
 See, e.g., Sections~$9.2$ and 9.7 of {\rm\cite{A11} for more details.
The two linearly independent solutions $\psi(x;k),\bar{\psi}(x;k)$ to \eqref{eqn:ZS1}
 are obtained via \eqref{eqn:a1a} from $\hat{\psi}(x;k),\hat{\psi}(x;-k)$ for $k\neq 0$.

\subsubsection{ZS system \eqref{eqn:ZS2}}

We turn to the ZS system \eqref{eqn:ZS2}
 and follow the standard IST theory for another class of integrable systems,
 which contains the examples of Section~1 except for the KdV equation \eqref{eqn:KdV}.
See, e.g., Chapter~2 of \cite{APT04} for more details of the theory.

Suppose that $a(k)$ and $\bar{a}(k)$ have $n$ and $\bar{n}$ simple zeros
 $\{k_j\}_{j=1}^n$ and $\{\bar{k}_j\}_{j=1}^{\bar{n}}$, respectively,
 where $\Im k_j>0$ and $\Im\bar{k}_j<0$.
Let
\[
N_j(x)=\psi(x;k_j)e^{-ik_jx},\quad
\bar{N}_j(x)=\bar{\psi}(x;\bar{k}_j)e^{i\bar{k}_jx}
\]
for $j=1,\ldots,n$ or $j=1,\ldots,\bar{n}$,
 where $\psi(x;k),\bar{\psi}(x;k)$ are the Jost solutions to \eqref{eqn:ZS2}
  satisfying \eqref{eqn:bc}.
Note that
\begin{equation}
\begin{split}
&
N_j(x)\sim
\begin{pmatrix}
0\\
1
\end{pmatrix},\quad
\bar{N}_j(x)\sim
\begin{pmatrix}
1\\
0
\end{pmatrix}\quad
\mbox{as $x\to+\infty$},\\
&
N_j(x)\sim b(k_j)^{-1}
\begin{pmatrix}
0\\
1
\end{pmatrix}e^{-2ik_jx},\quad
\bar{N}_j(x)\sim\bar{b}(\bar{k}_j)^{-1}
\begin{pmatrix}
1\\
0
\end{pmatrix}e^{2i\bar{k}_jx}\quad
\mbox{as $x\to-\infty$}
\end{split}
\label{eqn:N2a}
\end{equation}
since $\Im k_j>0$, $\Im\bar{k}_j<0$,
  $a(k_j),\bar{a}(\bar{k}_j)=0$ and $b(k_j),\bar{b}(\bar{k}_j)\neq 0$ by \eqref{eqn:det}.
We can show that they satisfy
\begin{equation}
\begin{split}
&
N_\ell(x)
=\begin{pmatrix}
0\\
1
\end{pmatrix}
+\sum_{j=1}^{\bar{n}}\frac{\bar{C}_j e^{-2i\bar{k}_jx}\bar{N}_j(x)}{k_\ell-\bar{k}_j},\quad
\ell=1,\ldots,n,\\
&
\bar{N}_\ell(x)
=\begin{pmatrix}
1\\
0
\end{pmatrix}
+\sum_{j=1}^n\frac{C_j e^{2ik_jx}N_j(x)}{\bar{k}_\ell-k_j},\quad
\ell=1,\ldots,\bar{n},
\end{split}
\label{eqn:N2}
\end{equation}
where
\begin{align*}
&
C_j=\frac{b(k_j)}{a_k(k_j)},\quad
j=1,\ldots,n,\\
&
\bar{C}_j=\frac{\bar{b}(\bar{k}_j)}{\bar{a}_k(\bar{k}_j)},\quad
j=1,\ldots,\bar{n}.
\end{align*}
Moreover, we have
\begin{equation}
q(x)=2i\sum_{j=1}^{\bar{n}}\bar{C}_j e^{-2i\bar{k}_jx}\bar{N}_{j1}(x),\quad
r(x)=-2i\sum_{j=1}^n C_j e^{2jk_jx}N_{j2}(x)
\label{eqn:qr2}
\end{equation}
and
\begin{equation}
\begin{split}
&
\psi(x;k)=\biggl(
\begin{pmatrix}
0\\
1
\end{pmatrix}
+\sum_{j=1}^{\bar{n}}\frac{\bar{C}_j e^{-2i\bar{k}_jx}\bar{N}_j(x)}{k-\bar{k}_j}\biggr)e^{ikx},\\
&
\bar{\psi}(x;k)=\biggl(
\begin{pmatrix}
1\\
0
\end{pmatrix}
+\sum_{j=1}^n\frac{\bar{C}_j e^{2ik_jx}N_j(x)}{k-\bar{k}_j}\biggr)e^{-ikx}
\end{split}
\label{eqn:ZS2sol}
\end{equation}
for $k\in\Cset$, where $N_{j\ell}(x)$ and $\bar{N}_{j\ell}(x)$
 are the $\ell$-th components of $N_j(x)$ and $\bar{N}_j(x)$, respectively.
See, e.g., Section~2.2.3 of \cite{APT04} for the derivations of the above relations.
Since they are obtained by the basic arithmetic operations from \eqref{eqn:N2},
 we see that $N_\ell(x)$ and $\bar{N}_\ell(x)$, $\ell=1,\ldots,n$ or $\bar{n}$,
 are rational functions of $e^{2ik_jx}$ and $e^{2i\bar{k}_jx}$, $j=1,\ldots,n$ or $\bar{n}$.
It follows from \eqref{eqn:N2a} and \eqref{eqn:qr2} that
\[
q(x),r(x)\to 0\quad\mbox{as $x\to\pm\infty$}.
\]
In the four examples \eqref{eqn:NLS}-\eqref{eqn:sinhG},
 Eq.~\eqref{eqn:qr2} corresponds to an initial condition of an $n$-soliton
 when $r(x)$ is appropriately defined with $n=\bar{n}$.
See, e.g., Section~$2.3$ of {\rm\cite{APT04} for more details.


\section{Proof of Theorem~\ref{thm:2a}}

In this section we prove Theorem~\ref{thm:2a}
 for \eqref{eqn:ZS1} and \eqref{eqn:ZS2} separately.

\subsection{ZS system \eqref{eqn:ZS1}}

We begin with the ZS system \eqref{eqn:ZS1}.
Henceforth we assume that the potential $q(x)$ is reflectionless
 and satisfies condition~{\rm(A)}.
We first prove the following.

\begin{lem}
\label{lem:5a}
$q(x)$ is a rational function of $e^{\lambda x}$ for some constant $\lambda>0$.
\end{lem}

\begin{proof}
Since $\hat{N}_\ell(x)$, $\ell=1,\ldots,n$, are rational functions
 of $e^{2ik_jx}$, $j=1,\ldots,n$, as stated in Section~4.2.1,
 we see that $q(x)$ is a rational function of $e^{2i k_jx}$, $j=1,\ldots,n_0$,
 after the order of $k_j$, $j=1,\ldots,n$, is changed if necessarily, where $1\le n_0\le n$.
If there does not exist a constant $\lambda>0$
 such that $k_j=i n_j\lambda$ with some integer $n_j>0$ for each $j=1,\ldots,n_0$,
 then $q(x)$ does not satisfy condition~(A) obviously.
Recall that $k_j$, $j=1,\ldots,n$, are purely imaginary in the upper half complex plane.
Thus, we obtain the result.
\end{proof}

Let $\hat{\psi}(x;k)$ be the Jost solution to the linear Schr\"odinger equation \eqref{eqn:ZS1a}
 satisfying \eqref{eqn:bc1a}, as in Section~4.2.1.

\begin{lem}
\label{lem:5b}
$\hat{\psi}(x;k)$ is a rational function of $e^{\lambda x}$ and $e^{ik x}$
 for $k\in\Cset\setminus\{0\}$.
\end{lem}

\begin{proof}
Let $\hat{q}(s)$ be a rational function of $s$ such that $q(x)=q_0(e^{\lambda x})$
 with $q_0(0)=0$ and $\lim_{s\to\infty}q_0(s)=0$.
The existence of such a rational function is guaranteed by Lemma~\ref{lem:5a}.
Using the transformation $s=e^{\lambda x}$, we rewrite \eqref{eqn:ZS1a} as
\begin{equation}
s^2w_{ss}+sw_s+\frac{k_j^2+q_0(s)}{\lambda^2}w=0
\label{eqn:lem5b}
\end{equation}
at $k=k_j$, $j=1,\ldots,n$.
Equation~\eqref{eqn:lem5b} is a linear differential equation over $\Cset(s)$
 and has regular singularities at $s=0$ and $\infty$.
See e.g., Section~7.1 of  \cite{CH11} for the definition of regular singularities
 in higher-order differential equations, which is similar
 to that in linear systems of first-order differential equations such as \eqref{LinearSystem}.
The indicial equations (e.g., Section~7.1 of  \cite{CH11}) at $s=0$ and $\infty$ coincide
 and are given by
\[
\rho^2+\frac{k_j^2}{\lambda^2}=0,
\]
which has two roots at
\[
\rho=\mp\frac{ik_j}{\lambda}:=\pm\rho_j\in\Rset,
\]
for $j=1,\ldots,n$.
Note that $-ik_j>0$, $j=1,\ldots,n$.

Assume that $\rho_j>0$ is not an integer.
Then the Jost solution $w=\hat{\psi}(x;k_j)=\hat{N}_j(x)e^{-ik_jx}$ to \eqref{eqn:ZS1a}
 corresponds to a solution to \eqref{eqn:lem5b}
 which converges to $w=0$ as $s\to 0$ and $\infty$, and has the forms
\begin{equation}
w=s^{\rho_j}w_1(s)
\label{eqn:lem5b1}
\end{equation}
near $s=0$ and
\begin{equation}
w=s^{-\rho_j}w_2(1/s)
\label{eqn:lem5b2}
\end{equation}
near $s=\infty$, where $w_\ell(s)$, $\ell=1,2$, are holomorphic functions of $s$
 (see, e.g., Section~7.1 of \cite{CH11}).
This yields a contradiction
 since if it has the form \eqref{eqn:lem5b1} near $s=0$
 then $\hat{N}_j(x)$ is a function of $e^{\lambda x}$,
 so that it does not have the form \eqref{eqn:lem5b2} near $s=\infty$.
Thus, for each $j=1,\ldots,n$,
 $\rho_j>0$ is an integer and $k_j=in_j\lambda$ with some $n_j\in\Nset$.
This implies that $\hat{N}_j(x)$, $j=1,\ldots,n$, are rational functions of $e^{\lambda x}$.
So the result immediately follows from \eqref{eqn:ZS1sol}.
\end{proof}

\begin{proof}[Proof of Theorem~$\ref{thm:2a}$ for \eqref{eqn:ZS1}]
The first part immediately follows from Lemma~\ref{lem:5a}.
We regard the ZS system \eqref{eqn:ZS1} as a linear system over $\Cset(e^{\lambda x})$.
Since $\hat{\psi}(x;k)$ and $\hat{\psi}(x;-k)$ are linearly independent solutions
 to the linear Schr\"odinger equation \eqref{eqn:ZS1a},
 we see via Lemma~\ref{lem:5b} that the Picard-Vessiot extension of \eqref{eqn:ZS1}
 is an exponential extension of $\Cset(e^{\lambda x})$.
Thus, we obtain the second part by Theorem~\ref{thm:dg}.
\end{proof}

\begin{rmk}
\label{rmk:5a}
If the ZS system \eqref{eqn:ZS1} is regarded as a linear system of differential equations
 over $\Cset(e^{2ik_1x},\ldots,e^{2ik_nx})$,
 then it is always integrable in the meaning of differntial Galois theory,
 although the potential $q(x)$ may not contain all of $e^{2ik_1x},\ldots,e^{2ik_nx}$.
\end{rmk}

\subsection{ZS system \eqref{eqn:ZS2}}

We turn to the ZS system \eqref{eqn:ZS2}.
Henceforth we assume that the potentials $q(x),r(x)$ are reflectionless
 and satisfy condition~{\rm(A)}.
We proceed as in Section~5.1.
We first prove the following like Lemma~\ref{lem:5a}.

\begin{lem}
\label{lem:5c}
$q(x),r(x)$ are rational functions of $e^{\lambda x}$
 for some constant $\lambda\in\Cset$ with $\Im\lambda>0$.
\end{lem}

\begin{proof}
Since $N_\ell(x)$, $\ell=1,\ldots,n$, and $\bar{N}_\ell(x)$, $\ell=1,\ldots,\bar{n}$,
 are rational functions of $e^{2ik_jx}$, $j=1,\ldots,n$,
 and $e^{2i\bar{k}_jx}$, $j=1,\ldots,\bar{n}$, as stated in Section~4.2.2,
 we see that $q(x),r(x)$ are rational functions
 of $e^{2ik_jx}$, $j=1,\ldots,n_0$, and $e^{2i\bar{k}_jx}$, $j=1,\ldots,\bar{n}_0$,
 after the orders of $k_j$, $j=1,\ldots,n$,
 and $\bar{k}_j$, $j=1,\ldots,\bar{n}$, are changed if necessarily,
 where $1\le n_0\le n$ and $1\le\bar{n}_0\le\bar{n}$.
If there does not exist a constant $\lambda\in\Cset$ with $\Re\lambda>0$
 such that $k_j=i n_j\lambda$ with some integer $n_j>0$ for $j=1,\ldots,n_0$
 and $\bar{k}_j=-i\bar{n}_j\lambda$ with some integer $\bar{n}_j>0$ for $j=1,\ldots,\bar{n}_0$,
 then $q(x),r(x)$ do not satisfy condition~(A) obviously.
Thus, we obtain the result.
\end{proof}

Let $\psi(x;k),\bar{\psi}(x;k)$ be the Jost solutions to the ZS system \eqref{eqn:ZS2}
 satisfying \eqref{eqn:bc}, as in Section~4.2.2.
We also prove the following like Lemma~\ref{lem:5b}.

\begin{lem}
\label{lem:5d}
$\psi(x;k),\bar{\psi}(x;k)$ are rational functions of $e^{\lambda x}$ and $e^{ik x}$
 for $k\in\Cset$.
\end{lem}

\begin{proof}
Let $q_0(s),r_0(s)$ be rational functions of $s$
 such that $q(x)=q_0(e^{\lambda x})$ and $r(x)=r_0(e^{\lambda x})$
 with $q_0(0),r_0(0)=0$ and $q_0(s),r_0(s)\to 0$ as $s\to\infty$.
The existence of such rational functions is guaranteed by Lemma~\ref{lem:5c}.
Using the transformation $s=e^{\lambda x}$, we rewrite \eqref{eqn:ZS2} as
\begin{equation}
v_s=\frac{1}{\lambda s}
\begin{pmatrix}
-ik & q_0(s)\\
r_0(s) & ik
\end{pmatrix}v.
\label{eqn:lem5d}
\end{equation}
Equation~\eqref{eqn:lem5d} is a linear system of differential equations over $\Cset(s)$
 and has regular singularities at $s=0$ and $\infty$.
 
Assume that $\rho_j=ik_j/\lambda$ is not an integer.
Noting that $\Im k_j>0$ and using Theorem~6 in Chapter~2 of \cite{B00},
 we see that the Jost solution $v=\psi(x;k_j)=N_j(x)e^{ik_j x}$ to \eqref{eqn:ZS2}
 corresponds to a solution to \eqref{eqn:lem5d} converge to $w=0$ as $s\to 0$ and $\infty$,
 and has the forms
\begin{equation}
v=s^{\rho_j}v_1(s)
\label{eqn:lem5d1}
\end{equation}
near $s=0$ and
\begin{equation}
v=s^{-\rho_j}v_2(1/s)
\label{eqn:lem5d2}
\end{equation}
near $s=\infty$, where $v_\ell(s)$, $\ell=1,2$, are vectors
 whose components are holomorphic functions of $s$.
This yields a contradiction
 since if it has the form \eqref{eqn:lem5d1} near $s=0$,
 then $N_j(x)$ is a function of $e^{\lambda x}$,
 so that it does not have the form \eqref{eqn:lem5d2} near $s=\infty$.
Thus, for each $j=1,\ldots,n$,
 $\rho_j=ik_j/\lambda$ is an integer and $k_j=in_j\lambda$ for some $n_j\in\Nset$.
Similarly, we can show that for each $j=1,\ldots,\bar{n}$,
 $i\bar{k}_j/\lambda$ is an integer and $\bar{k}_j=-i\bar{n}_j\lambda$ for some $\bar{n}_j\in\Nset$.
This implies that $N_j(x)$, $j=1,\ldots,n$, and $\bar{N}_j(x)$, $j=1,\ldots,\bar{n}$,
 are rational functions of $e^{\lambda x}$.
So the result immediately follows from \eqref{eqn:ZS2sol}.
\end{proof}

\begin{proof}[Proof of Theorem~$\ref{thm:2a}$ for \eqref{eqn:ZS2}]
The first part immediately follows from Lemma~\ref{lem:5c}.
We regard the ZS system \eqref{eqn:ZS1} as a linear system over $\Cset(e^{\lambda x})$.
Since $\hat{\psi}(x;k)$ and $\hat{\psi}(x;-k)$ are linearly independent solutions
 to the linear Schr\"odinger equation \eqref{eqn:ZS1a},
 we see via Lemma~\ref{lem:5b} that the Picard-Vessiot extension of \eqref{eqn:ZS1}
 is an exponential extension of $\Cset(e^{\lambda x})$.
Thus, we obtain the second part by Theorem~\ref{thm:dg}.
\end{proof}

\begin{rmk}
\label{rmk:5b}
If the ZS system \eqref{eqn:ZS2} is regarded as a linear system of differential equations
 over $\Cset(e^{2ik_1x},\ldots,e^{2ik_nx},e^{2i\bar{k}_1x},\ldots,e^{2i\bar{k}_{\bar{n}}x})$,
 then it is always integrable in the meaning of differntial Galois theory,
 although the potential $q(x),r(x)$ may not contain
 all of $e^{2ik_1x},\ldots,e^{2ik_nx},e^{2i\bar{k}_1x},\ldots,e^{2i\bar{k}_{\bar{n}}x}$
 $($cf. Remark~{\rm\ref{rmk:5a})}.
\end{rmk}


\section{Proof of Theorem~\ref{thm:2b}}

In this section we finally prove Theorem~\ref{thm:2b}.
Similar approaches were previously used
 to discuss nonintegrability and chaos
 in two-degree-of-freedom Hamiltonian systems in \cite{MP99,Y03,YY17,YY19}.

We first see that Eq.~\eqref{eqn:ZSpm} has a regular singularity at $s_\pm=0$
 since the matrices $A(s_\pm)$ are holomorphic.
Thus, we regard the ZS systems \eqref{eqn:ZS1} and \eqref{eqn:ZS2}
 as linear ODEs of Fuchs type on the Riemann surface $\hat{\Gamma}$.
Let $M_\pm$ be monodromy matrices of \eqref{eqn:ZSpm} around $s_\pm=0$.
Note that there exists no singularity on  $\hat{q}(U_R)$.
Let $\K=\{k\in\Rset\setminus\{0\}\mid ik(\lambda_+^{-1}-\lambda_-^{-1})\not\in\Zset\}$.
If $\lambda_+^{-1}-\lambda_-^{-1}\notin i\Rset$, then $\K=\Rset\setminus\{0\}$.

\begin{lem}
\label{lem:6a}
The monodromy matrices $M_\pm$ have eigenvalues
 $e^{2\pi k/\lambda_\pm}$ and $e^{-2\pi k/\lambda_\pm}$ for $k\in\K$.
\end{lem}

\begin{proof}
Let $k\in\K$.
Since $A_\pm(0)$ have eigenvalues $\pm ik$,
 the characteristic exponents of \eqref{eqn:ZSpm}
 are given by $\mp ik/\lambda_\pm$ and $\pm ik/\lambda_\pm$,
 the difference of which is not an integer.  
Hence, we compute the local monodromy matrices of \eqref{eqn:ZSpm} around $s_\pm=0$ as
\[
\exp\left(\mp\frac{2\pi i}{\lambda_\pm}A_\pm(0)\right),
\]
which have eigenvalues  $e^{2\pi k/\lambda_\pm}$ and $e^{-2\pi k/\lambda_\pm}$.
This means the desired result.
\end{proof}

Let $\Psi(x;k)$ be a fundamental matrix to \eqref{eqn:ZS1} or \eqref{eqn:ZS2}
 for $k\in\Rset\setminus\{0\}$.
Using a standard result about asymptotic behavior of linear ODEs
 (e.g., Section~3.8 of \cite{CL55}), we show that the limits
\[
B_\pm(k)=\lim_{x\to\pm\infty}\Phi(-x;k)\Psi(x;k)
\]
exist and $B_\pm(k)$ are nonsingular (cf. Lemma~3.1 of \cite{Y00}).
Recall that $\Phi(x;k)$ is a fundamental matrix to \eqref{eqn:ZS0} with $\Phi(0)=\id_2$
 and given by \eqref{eqn:Phi}.  
Hence, we have
\[
\Psi(x;k)\sim\Phi(x;k)B_\pm(k)\quad\mbox{as $x\to\pm\infty$}
\]
since $\Phi(x;k)^{-1}=\Phi(-x;k)$.
Letting
\[
B_0(k)=B_+(k)B_-(k)^{-1}\quad\mbox{and}\quad
\Psi_-(x)=\Psi(x)B_-(k)^{-1},
\]
we have
\begin{equation}
\begin{split}
&
\Psi_-(x;k)\sim\Phi(x;k)\quad\mbox{as $x\to-\infty$,}\\
&
\Psi_-(x;k)\sim\Phi(x;k)B_0(k)\quad\mbox{as $x\to+\infty$.}
\end{split}
\label{eqn:ab2}
\end{equation}
So the first and second column vectors of $\Psi_-(x;k)$
 give the Jost solutions $\phi(x;k)$ and $\bar{\phi}(x;k)$, respectively.
Similarly, the first and second column vectors of
\[
\Psi_+(x;k)=\Psi(x;k)B_+(k)^{-1}
\]
give the Jost solutions $\bar{\psi}(x;k)$ and $\psi(x;k)$, respectively.
From \eqref{eqn:ab1} and \eqref{eqn:ab2} we see that
\begin{equation}
B_0(k)=
\begin{pmatrix}
a(k) & \bar{b}(k)\\
b(k) & \bar{a}(k)
\end{pmatrix}.
\label{eqn:B0}
\end{equation}
Especially, $\det B_0(k)=1$ by \eqref{eqn:det}.

\begin{lem}
\label{lem:6b}
Let $k\in\K$.
The monodromy matrices can be expressed as 
\begin{equation}
M_+=B_0^{-1}
\begin{pmatrix}
e^{-2\pi k/\lambda_+} & 0\\
0 & e^{2\pi k/\lambda_+}
\end{pmatrix}B_0,\quad
M_-=
\begin{pmatrix}
e^{2\pi k/\lambda_-} & 0\\
0 & e^{-2\pi k/\lambda_-}
\end{pmatrix}
\label{eqn:lem6b}
\end{equation}
for a common fundamental matrix.
\end{lem}

\begin{proof}
Let $\tilde{\Psi}(x;k)=\Psi(x;k)B_-(k)^{-1}$.
Then $\tilde{\Psi}(x;k)$ is also a fundamental matrix to \eqref{eqn:ZS1} or \eqref{eqn:ZS2}
 such that
\[
\lim_{x\to-\infty}\Phi(-x)\tilde{\Psi}(x;k)=\id_2,\quad
\lim_{x\to-\infty}\Phi(-x)\tilde{\Psi}(x;k)=B_0.
\]
Consider the transformed ZS system consisting of \eqref{eqn:ZSU} and \eqref{eqn:ZSpm}
 on $\hat{\Gamma}$, and take a fundamental matrix corresponding to $\tilde{\Psi}(x;k)$.
Since its analytic continuation yields the (local) monodrmy matrices
\[
\begin{pmatrix}
e^{\mp 2\pi k/\lambda_\pm} & 0\\
0 & e^{\mp 2\pi k/\lambda_\pm}
\end{pmatrix}
\]
along small loops around $O_\pm$,
 which is estimated from asymptotic expressions
\[
T^{-1}\Phi\biggl(\mp\frac{1}{\lambda_\pm}\log s_\pm;k\biggr)T
\]
of its fundamental matrices,
 we choose the base point near $O_-$ to obtain the desired result.
\end{proof}

\begin{proof}[Proof of Theorem~$\ref{thm:2b}$]
Let $\M$ denote the monodromy group generated by $M_\pm$.
Assume that the hypothesis of Theorem~\ref{thm:2b} holds and $k\in\K$.
Then we have the following.

\begin{lem}
\label{lem:6c}
The monodromy group $\M$ is triangularizable.
\end{lem}

\begin{proof}
From Theorem~\ref{thm:sl} we first notice that
 $\M$ has the same classifications as stated in Proposition~\ref{prop:3a}.
So $\M$ is not an algebraic group of type (vi) in Proposition~\ref{prop:3a} obviously.
On the other hand, by Lemma~\ref{lem:6b}
 the eigenvalues of $M_\pm$ are not roots of $1$
 since $\lambda_\pm$ are not purely imaginary.
Hence, neither case (i), (ii) nor (iv) occurs for $\M$.
Thus, the monodromy group $\M$ is of type (iii) or (v).
\end{proof}

Theorem~\ref{thm:2b} is now easily proved.
Substituting \eqref{eqn:B0} into the first equation of \eqref{eqn:lem6b}, we have
\[
\begin{pmatrix}
a(k)\bar{a}(k)e_--b(k)\bar{b}(k)e_+
 & \bar{a}(k)\bar{b}(k)(e_--e_+)\\
 a(k)b(k)(e_+-e_-)
 & a(k)\bar{a}(k)e_+-b(k)\bar{b}(k)e_-
\end{pmatrix},
\]
where $e_\pm=e^{\pm 2\pi k/\lambda_+}$.
Hence, if the monodromy group $\M$ is triangularizable, then
\[
a(k)b(k)=0\quad\mbox{or}\quad
\bar{a}(k)\bar{b}(k)=0.
\]
Since by \eqref{eqn:ab3} $a_j(k)$, $j=1,2$, only have discrete zeros,
 we have $b(k)=0$ or $\bar{b}(k)=0$ for any $k\in\Rset\setminus\{0\}$
 by the identity theorem (e.g., Theorem~3.2.6 of \cite{AF03}).
This complete the proof by Theorem~\ref{thm:sl}. 
\end{proof}


\appendix

\renewcommand{\theequation}{A.\arabic{equation}}
\setcounter{equation}{0}

\section{Relations on scattering and reflection coefficients
 between \eqref{eqn:ZS1} and \eqref{eqn:ZS1a}}

Following Section~3d of \cite{N85} basically,
 we define the scattering and reflection coefficients for \eqref{eqn:ZS1a}
 (see also Section~9.1 of \cite{A11}).
Equation~\eqref{eqn:ZS1a} has the Jost solutions
\begin{equation}
\begin{split}
&
\hat{\phi}(x;k)\sim e^{-ikx}\quad\mbox{as $x\to-\infty$,}\\
&
\hat{\psi}(x;k)\sim e^{ikx}\quad\mbox{as $x\to+\infty$.}
\end{split}
\label{eqn:bc1a}
\end{equation}
We easily see that
\begin{equation*}
\begin{split}
&
\hat{\phi}(x;-k)\sim e^{ikx}\quad\mbox{as $x\to-\infty$,}\\
&
\hat{\psi}(x;-k)\sim e^{-ikx}\quad\mbox{as $x\to+\infty$.}
\end{split}
\end{equation*}
Hence, we have the relations
\begin{equation}
\begin{split}
&
\phi(x;k)=-\frac{i}{2k}\begin{pmatrix}
-\hat{\phi}_x(x;k)+ik \hat{\phi}(x;k)\\
\hat{\phi}_(x;k)
\end{pmatrix},\\
&
\bar{\phi}(x;k)=\begin{pmatrix}
-\hat{\phi}_x(x;-k)+ik \hat{\phi}(x;-k)\\
\hat{\phi}(x;-k)
\end{pmatrix},\\
&
\psi(x;k)=\begin{pmatrix}
-\hat{\psi}_x(x;k)+ik \hat{\psi}(x;k)\\
\hat{\psi}(x;k)
\end{pmatrix},\\
&
\bar{\psi}(x;k)=-\frac{i}{2k}\begin{pmatrix}
-\hat{\psi}_x(x;-k)+ik \hat{\psi}(x;-k)\\
\hat{\psi}(x;-k)
\end{pmatrix}
\end{split}
\label{eqn:a1a}
\end{equation}
between the Jost solutions to \eqref{eqn:ZS1} and \eqref{eqn:ZS1a} by \eqref{eqn:v}.
Define the scattering coefficients
 $\hat{a}(k)$ and $\hat{b}(k)$ for \eqref{eqn:ZS1a} as
\begin{equation}
\hat{\phi}(x;k)=\hat{a}(k)\hat{\psi}(x;-k)+\hat{b}(k)\hat{\psi}(x;k)
\label{eqn:ab1a}
\end{equation}
like \eqref{eqn:ab1}.
Since
\[
\hat{\phi}(x;-k)=\hat{a}(-k)\hat{\psi}(x;k)+\hat{b}(-k)\hat{\psi}(x;-k),
\]
 we have 
\begin{equation}
\hat{a}(k)\hat{a}(-k)-\hat{b}(k)\hat{b}(-k)=1
\label{eqn:a1b}
\end{equation}
like \eqref{eqn:det}.
From \eqref{eqn:a1a} we obtain
\begin{align*}
&
\phi(x;k)=\hat{a}(k)\bar{\psi}(x;k)-\frac{i}{2k}\hat{b}(k)\psi(x;k),\\
&
\bar{\phi}(x;k)=\hat{a}(-k)\psi(x;k)+2ik\hat{b}(-k)\bar{\psi}(x;k),
\end{align*}
which are compared with \eqref{eqn:ab1} to yield
\begin{equation}
a(k)=\hat{a}(k),\quad
\bar{a}(k)=\hat{a}(-k),\quad
b(k)=-\frac{i}{2k}\hat{b}(k),\quad
\bar{b}(k)=2ik\hat{b}(-k).
\label{eqn:a1c}
\end{equation}
Moreover, for the reflection coefficients we have
\begin{equation}
\rho(k)=-\frac{i}{2k}\hat{\rho}(k),\quad
\bar{\rho}(k)=2ik\hat{\rho}(-k),
\label{eqn:a1c}
\end{equation}
where $\hat{\rho}(k)=\hat{b}(k)/\hat{a}(k)$.


\end{document}